\documentclass[aps,pra,english,twocolumn, notitlepage,superscriptaddress,showpacs,nofootinbib]{revtex4-1}
\usepackage{mathptmx}

\usepackage[T1]{fontenc}
\usepackage[latin9]{inputenc}
\usepackage{color}
\usepackage{babel}
\usepackage{textcomp}
\usepackage{bm}
\usepackage{graphicx}
\usepackage{amsmath}
\usepackage{amsthm}
\usepackage{amssymb}
\usepackage{epstopdf}
\usepackage[unicode=true,pdfusetitle,
bookmarks=true,bookmarksnumbered=false,bookmarksopen=false,
breaklinks=true,pdfborder={0 0 0},pdfborderstyle={},backref=false,colorlinks=true]
{hyperref}
\hypersetup{
	urlcolor=blue, citecolor=blue, linkcolor=blue}

\makeatletter
\usepackage{amsmath, amsfonts, amssymb, amsthm, mathrsfs, amssymb, braket, bbm, bm,graphicx,times,txfonts,color,subfigure}
\usepackage{hyperref}

\newcommand{\tr}{\mathrm{Tr}}

\theoremstyle{plain}
\newtheorem{thm}{\protect\theoremname}
\theoremstyle{plain}

\theoremstyle{definition}

\theoremstyle{plain}
\newtheorem{lemma}{\protect\lemmaname}

\makeatother

\providecommand{\examplename}{Example}
\providecommand{\propositionname}{Proposition}
\providecommand{\theoremname}{Theorem}
\providecommand{\lemmaname}{Lemma}
\begin{document}

\title{Witnessing quantum coherence with prior knowledge of observables}

	\author{Mao-Sheng Li}
	\affiliation{ School of Mathematics,
		South China University of Technology, Guangzhou
		510641,  China}
	
	\author{Wen Xu}
	\affiliation{ School of Mathematics,
		South China University of Technology, Guangzhou
		510641,  China}	
	
\author{Shao-Ming Fei}
\email{feishm@cnu.edu.cn}
\affiliation{
School of Mathematical Sciences, Capital Normal University, Beijing 100048, China}
\affiliation{Max-Planck-Institute for Mathematics in the Sciences, 04103 Leipzig, Germany
}

\author{Zhu-Jun Zheng}\email{zhengzj@scut.edu.cn}
	\affiliation{ School of Mathematics,
		South China University of Technology, Guangzhou
		510641,  China}

	\author{Yan-Ling Wang}
\email{ylwang@dgut.edu.cn}
\affiliation{ School of Computer Science and Technology, Dongguan University of Technology, Dongguan, 523808, China}

\date{\today}

\begin{abstract}
Quantum coherence is the key resource in quantum technologies including faster computing, secure communication and advanced sensing. Its quantification and detection are, therefore, paramount within the context of quantum information processing. Having certain priori knowledge on the observables may enhance the efficiency of coherence detection. In this work, we posit that the trace of the observables is a known quantity. Our investigation confirms that this assumption indeed extends the scope of coherence detection capabilities. Utilizing this prior knowledge of the trace of the observables, we establish a series of coherence detection criteria. We investigate the detection capabilities of these coherence criteria from diverse perspectives and ultimately ascertain the existence of four distinct and inequivalent criteria.  These findings contribute to the deepening of our understanding of coherence detection methodologies, thereby potentially opening new avenues for advancements in quantum technologies.
\end{abstract}

\maketitle

\section{Introduction}

Quantum coherence \cite{Streltsov17,Hu2018} is a fundamental phenomenon in quantum mechanics, enabling quantum systems to exist in superpositions of multiple states simultaneously. It plays a pivotal role in various quantum technologies, from quantum computing \cite{Hillery16}, quantum metrology  \cite{Giovannetti04,Giovannetti11},  to nanoscale thermodynamics \cite{Gour15,Lostaglio15,Narasimhachar15,Francica19} and energy transport in biological systems \cite{Lloyd11,Huelga13,Romero14}. As such, coherence represents a critical resource for various quantum information processing tasks.

Significant strides have been made toward the measurement of quantum coherence based on numerous innovative approaches \cite{Baumgratz14,Aberg14,Yuan2015,Streltsov2015,Du2015,Winter2016,Napoli16,Piani16,Qi2017,Bu2017,Luo2017,Jin2018,Fang2018,Ringbauer18,Xi2018,Ma19,Xi2019,Bischof2019,Yu2020,Xu2020,Cui2020,Li2021,Bischof2021,Sun2021,Kim2021,Ray2022,Sun2022,Basiri2022,Sun-Yu2022,Yu-Yu2022,Liu2023}. Quantum resource theory \cite{Aberg14,Winter2016,Napoli16,Bu2017,Bischof2019} has garnered significant attention, proving instrumental in advancing our understanding of coherence. Under the quantum resource theory framework, various coherence monotones and measures have been introduced, including the relative entropy of coherence, the $\ell_1$ norm of coherence, the coherence of formation, the geometric measure of coherence and the robustness of coherence.

Coherence witnesses, akin to their entanglement witness \cite{Guhne09,Horodecki09,Wu07,Hou10a,Wang18} counterparts, have emerged as potent tools for coherence detection in experimental settings \cite{Ren2017,Wang17,Zheng18,Nie19} and coherence quantification in theoretical contexts. In contrast to conventional methods that rely on state tomography, coherence witnesses can directly identify the coherent states \cite{Ma2021,Wang21,Wang21b}. It is worth noting that Napoli et al. \cite{Napoli16} have made significant contributions to estimating the lower bound on the robustness of coherence by utilizing the coherence witnesses, emphasizing the importance of this concept in deepening our understanding of quantum coherence. More recently, Wang et al. \cite{Wang21,Wang21b} conducted a relatively comprehensive exploration of coherence witnesses. Nevertheless, research on coherence witnesses remains relatively limited. In practical applications, we often possess valuable prior knowledge about the observables involved. This naturally raises an important question: can such prior knowledge enhance the accuracy and effectiveness of coherence detection?

In this paper, we systematically investigate coherence witnesses within the framework of prior knowledge about observables. In Section \ref{sec:concept}, we provide an overview of coherence witnesses and unveiling a series of coherence criteria under varying degrees of prior knowledge concerning the traces of observables. In Section \ref{sec:property}, we analyze these coherence criteria from various perspectives, including completeness, finite completeness, finite intersection and inclusion. Finally, in Section \ref{sec:con} we summarize our findings and outline potential avenues for future researches.

\section{Coherence witnesses with prior knowledge of observables }\label{sec:concept}

We denote $\mathcal{H}$ a $d$-dimensional Hilbert space with computational basis $\mathcal{B}:=\{|i\rangle,\,i=1,2,\cdots,d\}$, $\mathbb{H}$ the set of all $d\times d$ Hermitian matrices, and $\mathbb{D}$ the set of all density matrices (self-adjoint positive semidefinite matrices with trace 1). Let $\Delta$ be the set of incoherent states with respect to the basis $\mathcal{B}$, namely, $\Delta$ comprises all density matrices of the form,
	\begin{equation}\label{ic}
		\delta=\sum_{i=1}^d\delta_i|i\rangle\langle i|.
	\end{equation}	
Hence, all the states within $\mathbb{D}\setminus \Delta$ are referred to as coherent states.

By definition the set $\Delta$ of incoherent states is convex and compact. From the Hahn-Banach theorem \cite{Edwards65}, there must exist a hyperplane which separates an arbitrary given coherent state from the set of incoherent states. A coherence witness is an Hermitian operator $W\in \mathbb{H}$ such that (i) $\tr(W\delta)\geq0$ for all incoherent states $\delta\in \Delta$, and (ii) there exists a coherent state $\rho$ such that $\tr(W\rho)<0$. The first condition implies that all the diagonals of $W$ are nonnegative, while the second implies that $W$ has some negative eigenvalue. We denote $\mathbb{H}_\geq$ the set of $d\times d$ Hermitian matrices with nonnegative diagonal elements and $\Lambda_-$ the set with some negative eigenvalue. The general coherence witnesses are then generally given by $\mathbb{H}_{\geq} \cap \Lambda_-.$

In addition to the above coherent witnesses, there are other kinds of coherent witnesses defined by strengthening the condition (i) to $\tr(W\delta)=0$ for all $\delta\in \Delta$, but relaxing the condition (ii) to $\tr(W\rho)\neq0$ for some state $\rho$ \cite{Napoli16,Ren2017,Wang21}. Here, the strengthen condition implies that all the     diagonals of $W$ must be zero, while the relaxed condition allows for detecting more coherent states by this witness. In other words, prior knowledge about the observables may enable one to detect better the coherence.

For a usual coherence witness $W\in \mathbb{H}_{\geq} \cap \Lambda_-$, the condition  $\tr(W\delta)=0$ for all incoherent states $\delta\in \Delta$ is equivalent to $\tr[W]=0.$
This motivates us to define the set of all coherence witnesses with the same trace. Therefore, for each real number $x\geq 0,$ we introduce the set:
$$\mathbb{W}_x:=\{W\in \mathbb{H}_\geq \ \big |\  \tr[W]=x\}.$$

The prior knowledge on the trace of witness can enhance the ability of detecting coherence.
For fixed $x>0$ we have for any $W=(w_{ij})\in \mathbb{W}_x$ and $\delta=\sum_{i=1}^d\delta_i|i\rangle\langle i|\in\Delta$,
$$
0\leq  \tr[W \delta]=\sum_{i=1}^d w_{ii} \delta_i \leq \sum_{i=1}^d w_{ii}=x.
$$
Therefore, if $W=(w_{ij})$ is priorly  known   in the set $\mathbb{W}_x$, then either $\tr[W\rho]<0$ or $\tr[W\rho]>x$ implies the coherence of $\rho$.  Moreover, there do exist some $W=(w_{ij})\in \mathbb{W}_x$ and density matrix $\rho$ such that  $\tr[W\rho]>x.$  In fact, for any $0<r<x$, let
$$
W=(x-r)|1\rangle \langle 1|+\sqrt{rx}\big( |1\rangle \langle 2|+|2\rangle \langle 1| \big) + r|2\rangle \langle 2|
$$
and $W=\lambda_+|v_+\rangle\langle v_+|+\lambda_- |v_-\rangle\langle v_-|$ be its spectral decomposition. As $\lambda_+\lambda_-=(x-r)r-(\sqrt{rx})^2=-r^2<0$, we can always assume $\lambda_+>0$ and $\lambda_-<0$. For $\rho=|v_+\rangle\langle v_+|$ we have
$$\tr[W\rho]=\lambda_+ =(\lambda_++\lambda_-)-\lambda_-=\tr[W]-\lambda_-=x-\lambda_- >x.$$
As a consequence, the prior knowledge on the trace of the observables indeed extends the scope of coherent state detections (see Fig. \ref{fig:compare}).

	\begin{figure}[ht]
	\centering
	\includegraphics[scale=0.34]{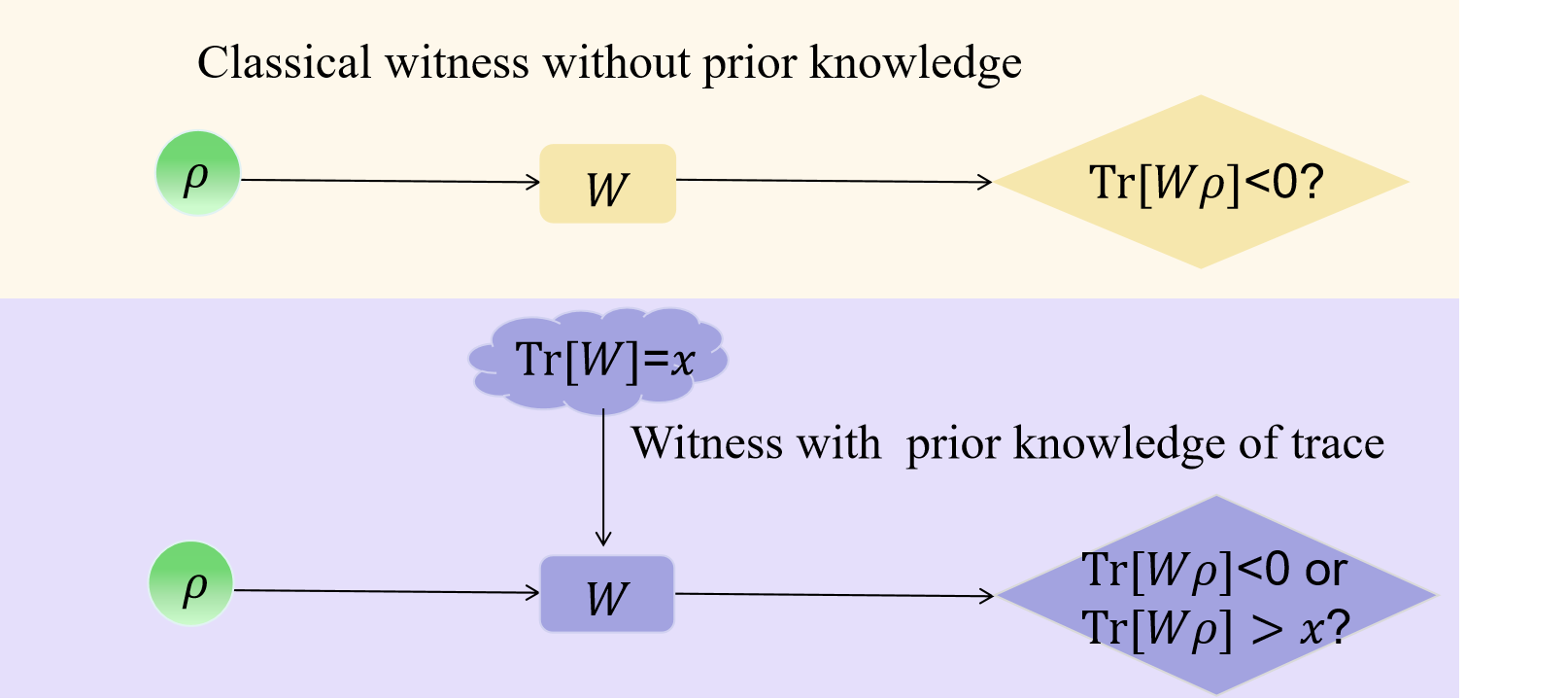}	
	\caption{ Classical coherence witness vs coherence witness with prior knowledge of trace. }\label{fig:compare}
\end{figure}

However, if we just know that   the observables could not be traceless, i.e., such observables are choosing from the set   $\mathbb{W}_>:=\{W\in \mathbb{H}_\geq\ \big |\ \tr[W]>0\}$, one finds that such little  prior knowledge of the trace of  observables could not help to enhance the coherence detection.  That is, 
to ensure that the witness $W$ detects the coherence of $\rho$, we still need to observe the violation of $\tr[W\rho]\geq 0$ as usual.  For proving this, we only need to show that 
 $$ \{\tr[W \delta] \ \big| \ W\in \mathbb{W}_>, \delta\in\Delta\}=\{x\in \mathbb{R} \mid x\geq 0\},$$
 that is, for any real number $x\geq 0$, there always exist $W\in \mathbb{W}_>$ and $\delta \in \Delta$ such that
$\tr[W\delta]=x$. If $x>0$, choosing any $W\in \mathbb{W}_>$ with $\tr[W]=dx$ and $\delta=\frac{\mathbb{I}_d}{d} $, we have
$$\tr[W\delta]=\frac{1}{d} \tr[W]=x.$$
If $x=0$, setting $W=|1\rangle\langle 1|+|1\rangle\langle 2| +|2\rangle\langle 1|$ and $\rho=|2\rangle\langle 2|$, we have  $\tr[W\rho]=0.$

Throughout this paper, we set $\mathcal{X}:=\{r\in\mathbb{R}|r\geq 0\} \cup \{>,\geq\}.$ For each $x\in \mathcal{X},$ we have set $I_x:=\{\tr[W \delta] \ \big| \ W\in \mathbb{W}_x, \delta\in\Delta\}$ and $D_x:=\{\tr[W \rho] \ \big| \ W\in \mathbb{W}_x, \rho\in\mathbb{D}\}.$  For each $x\in \mathcal{X}$, a corresponding coherence witness criterion can be formulated as follows.

\vskip 6pt

\noindent  {\bf Coherence Criterion: $(\mathbb{W}_x,\ D_x,\ I_x)$.} For any $\rho\in \mathbb{D},$ if  there exists some $W\in \mathbb{W}_x,$ such that $\tr[W\rho]\in D_x\setminus I_x,$  then $\rho$ is a coherent state.

For any $x\in \mathcal{X},$ it is directly to verify that $D_x=\mathbb{R}$. From the above discussion, the set $I_x$ is classified into following three classes, see Fig. \ref{fig:structureT}:
\begin{enumerate}
	\item [{\rm(a)}] $I_x=\{r\in\mathbb{R}\ \big| 0\leq   r\leq x\}$ when $x$ is a positive real number;
	\item [{\rm(b)}] $I_x=\{r\in\mathbb{R}\ \big| \ r\geq 0\}$ when $x$ is either $\geq$  or $>$;
	\item [{\rm(c)}] $I_0=\{0\}$ when $x=0$.
\end{enumerate}
	\begin{figure}[ht]
			\centering
			\includegraphics[scale=0.7]{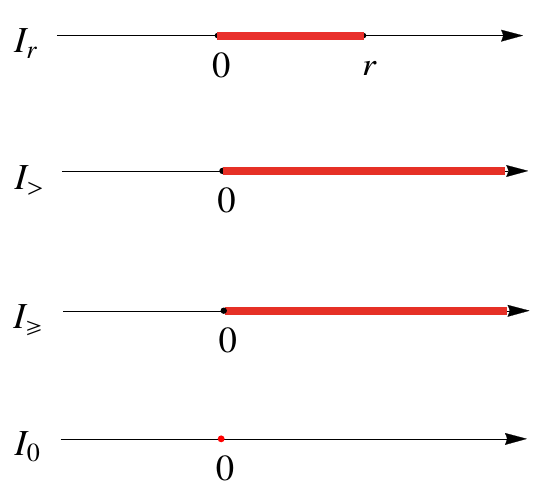}	
			\caption{ An intuitive  view of the range of $I_x$ for $x\in \{r\in\mathbb{R}|r\geq 0\} \cup \{>,\geq\}$. }\label{fig:structureT}
		\end{figure}

To classify coherence criteria we define $\mathbb{C}_x[W]$ to be the set of all coherent states that can be witnessed by $W$ in the setting $(\mathbb{W}_x,\ D_x,\ I_x)$, that is, $\mathbb{C}_x[W]=\{ \rho\in \mathbb{D} \ | \ \tr[W\rho]\in D_x\setminus I_x\}.$  Moreover, we denote $\mathbb{W}_x^c:=\{W\in \mathbb{W}_x\ \big | \ \mathbb{C}_x[W]\neq \emptyset\}$, i.e., the set of all nontrivial coherence witnesses of coherence criterion $(\mathbb{W}_x, D_x,I_x)$. One finds that  $\mathbb{W}_x^c= \mathbb{W}_x \cap \Lambda_-$ (see the proof in appendix  \ref{appen:lemma}).

For $W\in \mathbb{H}_\geq,$  if we set $x=\tr[W]$, then  $W$ belongs to $\mathbb{W}_x$  and   $\mathbb{W}_\geq.$ One finds that
$$\mathbb{C}_x[W]= \mathbb{C}_\geq[W] \cup \mathbb{C}_\geq [x\mathbb{I}_d-W].$$ It is easy to check that  $\mathbb{C}_\geq[W]$ is a convex set, i.e., if   $\tr[W\rho]<0$ and  $\tr[W\sigma]<0$, then $ \tr[W \left( t \rho +(1-t)\sigma \right)]<0$ for $t\in [0,1].$  However,    if $\mathbb{C}_\geq[W] \neq \emptyset $ and   $\mathbb{C}_\geq [x\mathbb{I}_d-W] \neq \emptyset $, the set $\mathbb{C}_x[W]$ is not a convex set. In this case, $\mathbb{C}_x[W]$ is a disjoint union of two convex sets, see cases (A) and (C) of  Fig. \ref{fig:convex}. In fact, for any $\rho \in \mathbb{C}_\geq[W] $ and $\sigma \in  \mathbb{C}_\geq [x\mathbb{I}_d-W] $, we must have $\tr[W\rho]<0$ and $\tr[W\sigma]>x\geq 0$. Therefore, there are must some $t_*\in(0,1)$ such that $ \tr[W \left( t_* \rho +(1-t_*)\sigma \right)]=0$, which indicates that $ t_* \rho +(1-t_*)\sigma  \notin \mathbb{C}_x[W].$
\begin{figure}[ht]
	\centering
	\includegraphics[scale=0.31]{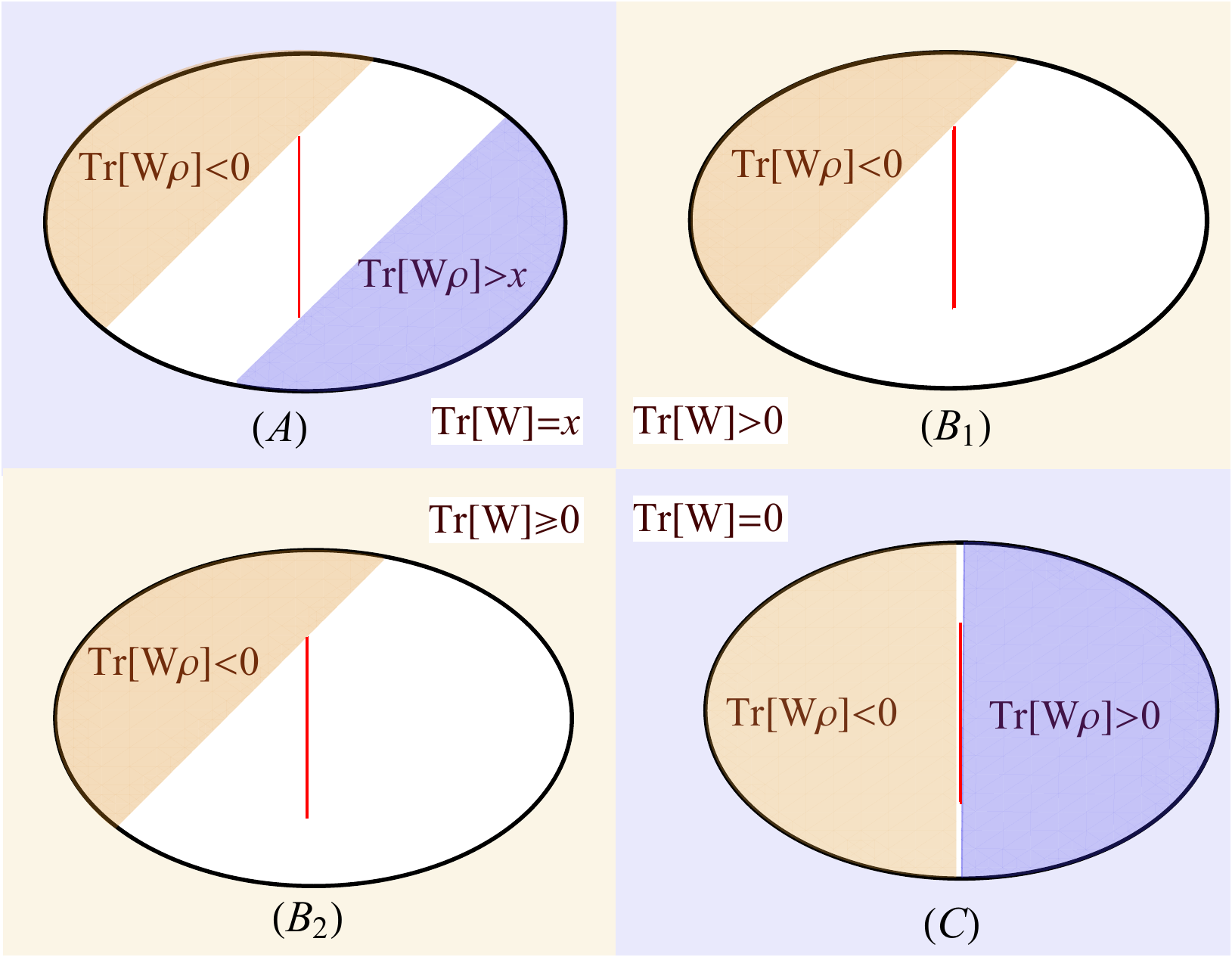}	
	\caption{ An intuitive view of the detection range of each type of coherence witnesses. Here the red line represents the set of incoherent states. The deep orange area represents the coherent states that could be detected by classical witness, i.e.,  $\mathbb{C}_\geq [W]$. The deep purple area represents the coherent states that could be detected by  the condition $\tr[W\rho]>x,$ namely, $\mathbb{C}_\geq [x\mathbb{I}_d-W]$. }\label{fig:convex}
\end{figure}

We say that two coherence criteria $(\mathbb{W}_x, D_x,I_x)$ and $(\mathbb{W}_y, D_y,I_y)$ are equivalent if there is a bijection $F: \mathbb{W}_x^c\rightarrow \mathbb{W}_y^c$ such that $\mathbb{C}_x[W]= \mathbb{C}_y[F(W)]$ for any $W\in \mathbb{W}_x^c.$ It is an intuition that the coherence criteria arising from knowing different values of traces may be almost the same in detecting coherence.
  For any positive real number $x,y$, the two criterions $(\mathbb{W}_x, D_x,I_x)$ and $(\mathbb{W}_y, D_y,I_y)$ are   equivalent in the above sense.
In fact, we  can define a map $F$ from $\mathbb{W}_x^c$ to $\mathbb{W}_y^c$  by sending $W\in \mathbb{W}_x^c$ to  $\frac{y}{x} W.$ Clearly, $\frac{y}{x} W\in \mathbb{W}_y^c$ as $W$
has nonnegative diagonals and $\tr[\frac{y}{x} W]=\frac{y}{x}\tr[W]=\frac{y}{x}x=y.$ It is easy to verify that it is a bijection. For any $\rho\in \mathbb{D},$ $0\leq \tr[W\rho]\leq x$ is equivalent to  $0\leq \tr[\frac{y}{x}W\rho]\leq y.$ Therefore, we always have $\mathbb{C}_x[W]= \mathbb{C}_y[F(W)].$

However, to fully classify the previous mentioned coherence criteria, we need to study more properties on these criteria. After that, we will go back to  tackle with this   problem.

\section{Properties and classification of coherence criteria}\label{sec:property}

For each $x\in \mathcal{X}$,  we have established  the coherence criterion  $(\mathbb{W}_x,\ D_x,\ I_x)$. It is naturally to ask whether these criteria are complete in the sense that they detect all the coherent states, namely, the following relation holds,
$$
\mathbb{D}\setminus \Delta= \bigcup_{W\in \mathbb{W}_x} \mathbb{C}_x[W].
$$
Concerning the completeness of our coherence criteria, we have the following conclusion whose proof can be seen in appendix \ref{appen:thm1}.

\begin{thm}[Completeness]\label{thm:On_Completeness}
For each  $x\in \mathcal{X},$ the coherence criterion ($\mathbb{W}_x, D_x,I_x$) is complete.
\end{thm}

So each coherence criterion presented here is complete.
As consequences, we can restate the coherence criteria by the following four classes (see also Fig. \ref{fig:convex}):
\begin{enumerate}
	\item [{\rm(A):}] $(\mathbb{W}_r, D_r,I_r)$ with $r$ a positive real number. A state $\rho\in \mathbb{D}$ is coherent if and only if there exists $W\in \mathbb{W}_r$ such that $\tr[W\rho]< 0$ or $\tr[W\rho]> r.$
	\item [{\rm($\mathrm{B}_1$):}] $(\mathbb{W}_>, D_>,I_>)$. A state $\rho\in \mathbb{D}$   is coherent if and only if there exists $W\in \mathbb{W}_>$ such that $\tr[W\rho]< 0.$
	\item [\rm($\mathrm{B}_2$):] $(\mathbb{W}_\geq, D_\geq,I_\geq)$. A state $\rho\in \mathbb{D}$ is coherent if and only if there exists $W\in \mathbb{W}_\geq$ such that $\tr[W\rho]< 0.$
	\item [{\rm(C):}]$(\mathbb{W}_0, D_0,I_0)$. A state $\rho\in \mathbb{D}$ is coherent if and only if there exists $W\in \mathbb{W}_0$ such that $\tr[W\rho]\neq 0.$
\end{enumerate}

A coherence criterion ($\mathbb{W}_x, D_x,I_x$) is called finitely completable if all the coherent states could be detected by a finite set of coherent witnesses in $\mathbb{W}_x^c.$ That is, there exists a finite set $\{W_i\}_{i=1}^n\subseteq \mathbb{W}_x^c$ such that
$$ \mathbb{D}\setminus \Delta=\bigcup_{i=1}^n \mathbb{C}_x[W_i]. $$
Otherwise, we call it finitely incompletable.

\begin{thm}[Characterization of finite completeness]\label{thm:Characterization_Completeness}
Let $x$ be an element in $\mathcal{X}.$ The coherence criterion ($\mathbb{W}_x, D_x,I_x$) is   finitely completable if and only if   $x$ is  $0$ or  $\geq$.
\end{thm}

See appendix \ref{appen:thm2}  for the proof. From Theorem \ref{thm:Characterization_Completeness}, only those criteria    that include all observables with traces   zero  are finitely completable. From this sense, the types $(A)$ and  $(B_1)$ are quite different from the types $(B_2)$ and $(C).$  Although  the criteria  or the figures  (presented in Fig. 
\ref{fig:convex}) of types $(B_1)$ and $(B_2)$  seem the same, they are in fact inequivalent.

 In the following we discuss the conditions when the finite intersections of $\mathbb{C}_x[{W}]$ maybe empty, for which the case with $x$ being $\geq$ has been discussed in Ref. \cite{Wang21}. 

\begin{thm}[Property of finite intersection]\label{thm:Common_State}
Let $x\in \mathcal{X}$ and $\mathcal{W}:=\{W_i\}_{i=1}^n\subseteq \mathbb{W}_x^c.$  Then   the following statements hold.
\begin{enumerate}
\item [{\rm(a)}]For the case $x\in (0,\infty)$ and any subset $\mathcal{W}'$ with cardinality $|\mathcal{W}'|\geq \lceil\frac{n}{2}\rceil$, if there always exist $t_W >0$ for each $W\in \mathcal{W}'$ and $\sum_{W\in \mathcal{W}'} t_W=1$ such that $  \sum_{W\in \mathcal{W}'}  t_W W \geq \boldsymbol{0}$, then  $ \bigcap_{i=1}^n \mathbb{C}_x[{W_i}] = \emptyset.$
  \item [{\rm(b)}] If $x$ is $>$ or $\geq,$  then
      $ \bigcap_{i=1}^n \mathbb{C}_x[{W_i}] = \emptyset
$ if and only if there is some nonnegative  $t_i\geq 0$ and $\sum_{i=1}^n t_i=1$ such that $\sum_{i=1}^n t_iW_i$ is positive semidefinite.
 \item [\rm(c)] If $x=0,$ there exists a common coherent state that can be detected by all of $W_i$ via criterion of type (C). That is, we always have
    $$
\mathcal{C}_0[{\mathcal{W}}]:=\bigcap_{i=1}^n \mathbb{C}_0[{W_i}] \neq \emptyset.
$$
Moreover, the cardinality of $\mathcal{C}_0[{\mathcal{W}}]$ is infinity. That is, given any set of finitely many coherent witnesses,  they     share with infinitely many common coherent states.
\end{enumerate}
\end{thm}

The proof is in  appendix \ref{appen:thm3}. From  Theorem \ref{thm:Common_State}, we conclude that the coherence criterion of type (C) is different from the other types.  Moreover,  we can also conclude that the coherence criterion of type (A) is also different from type (B) (that is, type $(\mathrm{B}_1)$ and  type $(\mathrm{B}_2)$) by the following examples.

Firstly,  there do exist $W_i$'s $\in \mathbb{W}_x^c$ which  satisfy
 all the conditions of (a) in Theorem \ref{thm:Common_State}. Take $d=n=3 $ for example,
 $$
 W_1=\left[\begin{array}{ccc}
  \frac{x}{2}&  0 & \frac{x}{6}\\[2mm]
    0&  \frac{x}{2} & \frac{x}{6}\\[2mm]
        \frac{x}{6}& \frac{x}{6} & 0\\[2mm]
 \end{array}\right],~
  W_2=\left[\begin{array}{ccc}
  0&  \frac{x}{6} & \frac{x}{6}\\[2mm]
    \frac{x}{6}&  \frac{x}{2} & 0\\[2mm]
        \frac{x}{6}& 0 & \frac{x}{2}\\[2mm]
 \end{array}\right],~
   W_3=\left[\begin{array}{ccc}
  \frac{x}{2}&  \frac{x}{6} &  0\\[2mm]
    \frac{x}{6}&  0 &  \frac{x}{6}\\[2mm]
       0&  \frac{x}{6} & \frac{x}{2}\\[2mm]
 \end{array}\right].
 $$
 One easily checks that $\frac{1}{2}W_1+ \frac{1}{2}W_2,$ $\frac{1}{2}W_2+ \frac{1}{2}W_3,$ $\frac{1}{2}W_1+ \frac{1}{2}W_3$ and $\frac{1}{3}W_1+ \frac{1}{3}W_2 +\frac{1}{3}W_3$ satisfy the assumed conditions. Hence $\mathbb{C}_x[W_1] \cap  \mathbb{C}_x[W_2]\cap  \mathbb{C}_x[W_3]=\emptyset.$ Moreover, for the case $d=n=2$, if we set
 $$
 W_1=\left[\begin{array}{cc}
 x & -\frac{x}{2}\\[2mm]
 -\frac{x}{2}& 0
 \end{array}\right],\
  W_2=\left[\begin{array}{cc}
 0 &  \frac{x}{2}\\[2mm]
  \frac{x}{2}& x
 \end{array}\right],\  \rho=\left[\begin{array}{cc}
 \frac{1}{4} &  \frac{1}{3}\\[2mm]
  \frac{1}{3}& \frac{3}{4}
 \end{array}\right],
 $$
 then $\frac{1}{2}W_1+\frac{1}{2}W_2=x \mathbb{I}_2.$ Hence $\mathbb{C}_>[W_1] \cap  \mathbb{C}_>[W_2] =\mathbb{C}_\geq[W_1] \cap  \mathbb{C}_\geq[W_2]=\emptyset.$   However, $\tr[W_1\rho]= -\frac{x}{12}<0$ and $\tr[W_2\rho]= \frac{13}{12}x>x.$ Therefore, $\rho\in \mathbb{C}_x[W_1] \cap  \mathbb{C}_x[W_2].$  
 
 Now we are studying more  about the  inclusion  and identity relations  among the family of sets $\mathbb{C}_x[{W}]$ where $W\in  \mathbb{W}_x$.  

\begin{thm}[Property of inclusion or identity]\label{thm:inclusion_State}
Let $x\in \mathcal{X}$ and $W_1, W_2\in \mathbb{W}_x^c.$ The following statements hold.
\begin{enumerate}
\item [{\rm(a)}] If $x\in (0,\infty),$ then $\mathbb{C}_x[{W_1}]=\mathbb{C}_x[{W_2}]$ if and only if $W_1=W_2$ or $W_1+W_2=x \mathbb{I}_d$ (the statement is true only if $d=2$).
  \item [{\rm(b)}]  If  $x$ is $>$ or $\geq$, $\mathbb{C}_x[{W_1}]=   \mathbb{C}_x[{W_2}]$ if and only if there exists $r\in\mathbb{R}_+:=\{r\in\mathbb{R} \ \big |\ r>0\}$ such that  $W_2= r W_1.$  Moreover,  $\mathbb{C}_x[{W_2}]\subseteq    \mathbb{C}_x[{W_1}]$  if and only if  there exists $a\in \mathbb{R}_+$ and a positive semidefintie operator $P$ such that $W_2= a W_1+P$.
 \item [\rm(c)] If $x=0,$ then $\mathbb{C}_0[W_1]=\mathbb{C}_0[W_2] $ if and only if there exist some $r\in \mathbb{R}\setminus\{0\}$ such that $W_2=rW_1$ (we call such two witnesses $\mathbb{R}$ equivalent). If $W_1$ and $W_2$ are not $\mathbb{R}$ equivalent, then $\mathbb{C}_0[W_1]\not\subseteq \mathbb{C}_0[W_2]$ and $\mathbb{C}_0[W_2]\not\subseteq \mathbb{C}_0[W_1].$
\end{enumerate}
\end{thm}

 The proof is presented in appendix \ref{appen:thm4}.
With the above results, we are ready to give a full classification of the variant coherence criteria.

\vskip 5pt

\begin{thm}\label{thm:Class_Coherence_Criteroin}
There are exactly four inequivalent classes
among the variant coherence criteria $\{(\mathbb{W}_x, D_x,I_x)\ \big| \  x\in  \mathcal{X}\}.$
\end{thm}
\begin{proof}
Firstly, we have shown that the criteria of type (A) are all equivalent.

Secondly, the criterion of type (C) is not equivalent to other types. In fact, in the settings of (A), ($\mathrm{B}_1$) and ($\mathrm{B}_2$), there are always two witnesses $W_1$ and $W_2$ such that $\mathbb{C}_x[W_1]\cap \mathbb{C}_x[W_2]=\emptyset.$ However, this is not the case when $x=0$.

Thirdly, criterion of type ($\mathrm{B}_2$), i.e., $(\mathbb{W}_\geq , D_\geq ,I_\geq)$, is not equivalent to the types (A) and ($\mathrm{B}_1$). In fact, criterion of type ($\mathrm{B}_2$) is finitely completable but the other twos not.

Moreover, $\mathbb{C}_>[W]$ are always convex for each $W\in \mathbb{W}_>^{c}$. But this fact may not be true for $\mathbb{C}_x[W]$ when $x>0$. Therefore, criteria of type (A) and ($B_1$) are inequivalent. To conclude we have exactly four inequivalent classes of coherence criteria. \end{proof}

\section{Conclusion and Discussions}\label{sec:con}

Based on prior knowledge of observables, we have presented a series coherence criteria which detect coherence better than the usual ones without prior knowledge of observables. Moreover, through a systematic and rigorous study on the  properties of criteria such as completeness, finite completeness, finite intersection and inclusion, we have singled out four classes of inequivalent coherence criteria. These results help to the deepening of our understanding of coherence detection methodologies, and thereby highlight on advancements in quantum technologies.

There are also some interesting problems left such as the condition of inclusion of $\mathbb{C}_x[W]$ when $x>0$ for the part (a) of Theorem \ref{thm:inclusion_State}. Further exploration of coherence detection enriched by prior knowledge may promise to unlock even greater potential. It would be also appealing to extend our scheme to deal with the case of entanglement witnesses. Our research may serve as a catalyst for future investigations on quantum coherence detection, as well as detections of other resources like quantum correlations.

\begin{acknowledgements}
This work is supported by National Natural Science Foundation of China (Grant Nos. 12371458, 62072119, 12075159, 12171044), the Guangdong Basic and Applied Basic Research Foundation under Grants No. 2023A1515012074, Key Research, Development Project of Guangdong Province under Grant No.2020B0303300001, Key Lab of Guangzhou for Quantum Precision Measurement under Grant No.202201000010, the Academician Innovation Platform of Hainan Province, and the Science and
Technology Planning Project of Guangzhou under Grants No. 2023A04J1296.
\end{acknowledgements}

\appendix
\setcounter{equation}{0}
\setcounter{table}{0}
\setcounter{section}{0}
\setcounter{lemma}{0}
\setcounter{thm}{0}

\section{Two  Lemmas  and a proof  of  the claim: $\mathbb{W}_x^c= \mathbb{W}_x  \cap \Lambda_-$
 }\label{appen:lemma}

In order to prove the claim  $\mathbb{W}_x^c= \mathbb{W}_x  \cap \Lambda_-$, we need the following lemma. 

\begin{lemma}\label{lemma:FindPostive}
	If $W\in \mathbb{H}_\geq \setminus \{\mathbf{0}\}$ is not positive semidefinite, then there exists some density matrices $\rho\in\mathbb{D}$ such that $\tr[W\rho]<0$. Moreover, there also exists positive density matrix $\sigma \in\mathbb{D}$ such that $\tr[W\sigma]=0.$
\end{lemma}

\begin{proof}
	By the given assumption, the matrix $W$ has both positive and negative eigenvalues. Suppose the spectral decomposition of $W$ is
	$$
	W=\sum_{\lambda_i >0} \lambda_i  |e_i\rangle\langle e_i| + \sum_{\lambda_j<0} \lambda_j  |e_j\rangle\langle e_j| +\sum_{\lambda_k =0} \lambda_k  |e_k\rangle\langle e_k|,$$
	where $\{|e_i\rangle\}_{\lambda_i>0}\cup  \{|e_j\rangle\}_{\lambda_j<0}\cup \{|e_k\rangle\}_{\lambda_k=0} $ is an orthonormal basis of the system $\mathcal{H}.$ For any fixed $j$ such that $\lambda_j<0,$  the state $\rho=|e_j\rangle\langle e_j|$ satisfies $\tr[W \rho]=\lambda_j<0.$ Set $\alpha:=\sum_{\lambda_i >0} \lambda_i$  and  $\beta:=\sum_{\lambda_j <0} |\lambda_j|.$ Clearly, $\alpha,\beta>0.$  We define
	$$ Q=\sum_{\lambda_i >0} \beta   |e_i\rangle\langle e_i| + \sum_{\lambda_j<0} \alpha  |e_j\rangle\langle e_j| +\sum_{\lambda_k =0}    |e_k\rangle\langle e_k|.$$
	Then $Q$ is a positive matrix. Moreover, we have $$\tr[WQ]= \beta\sum_{\lambda_i>0}  \lambda_i+\alpha  \sum_{\lambda_j<0} \lambda_j=\beta \alpha +\alpha \times (-\beta)=0.$$
	Then $\sigma=Q/\tr[Q]$ is a density matrix we wanted. That is, $\sigma$ is a positive matrix and satisfies $\tr[W\sigma]=0.$
\end{proof}

\vskip 5pt

\noindent {\bf The proof of the claim: $\mathbb{W}_x^c= \mathbb{W}_x  \cap \Lambda_-$.} Clearly, both sets are contains  in $\mathbb{W}_x\setminus \{\mathbf{0}\}.$
For any $W\in \mathbb{W}_x  \cap \Lambda_-$,  $W$ has some negative eigenvalue. By Lemma \ref{lemma:FindPostive}, there exists $\rho$ such that $\tr[W\rho]<0$. Hence, $\mathbb{C}_x[W]\neq \emptyset$ for such $W.$ 
 Therefore, $  \mathbb{W}_x  \cap \Lambda_- \subseteq \mathbb{W}_x^c$. 
 
 Now suppose that $W\in \mathbb{W}_x\setminus \{\mathbf{0}\}$ but $ W\notin  \mathbb{W}_x  \cap \Lambda_-$. So $W$ must be positive semidefinite (which is impossible when $x=0$) which implies that $\tr[W\rho]\geq 0$  for all $\rho\in\mathbb{D}.$ Therefore, when $x$ is $\geq$ or $>$,  $\mathbb{C}_x[W]=\emptyset.$ If $x$ is a positive real number, all the eigenvalues of $W$ must be less or equal to $x.$ Hence we have $\boldsymbol{0}\leq W \leq x \mathbb{I}_d.$ Therefore, $0\leq \tr[W \rho]\leq x$ for all $\rho\in\mathbb{D}$ and   $\mathbb{C}_x[W]=\emptyset$ for such $W.$  No matter what $x$ is, the set $\mathbb{C}_x[W]=\emptyset$. This implies that such $W$ is not an effective coherence witness, i.e., $W\notin \mathbb{W}_x^c.$ Hence we must have  $  \mathbb{W}_x  \cap \Lambda_- = \mathbb{W}_x^c$. \qed

\vskip 5pt

In the following, we present another lemma that will be reused  multiple times  throughout this paper.
 
\begin{lemma}\label{lemma:spaceEqual}
	Let $x\in \mathcal{X}$ and $W\in \mathbb{W}_x^c$ be a coherent witness. Define $S_W:=\{\rho\in \mathbb{D} | \ \tr[W\rho]=0\}$, $H_W:=\{H\in \mathbb{H} | \ \tr[W H]=0\}$ and $M_W:=\{X\in \mathrm{Mat}_d(\mathbb{C}) | \ \tr[WX]=0\}.$ Then the following sets are equal,
	$\mathrm{span}_{\mathbb{C}}(S_W)= \mathrm{span}_{\mathbb{C}}(H_W)=M_W$.
\end{lemma}
\begin{proof}
	Notice that each $W\in \mathbb{W}_x^c$  is not positive semidefinite. By Lemma \ref{lemma:FindPostive}, there is a positive state $\rho\in \mathbb{D}$ such that $\tr[W\rho]=0.$ So $\rho\in S_W$ by definition.
	
	Clearly, $S_W\subseteq H_W\subseteq M_W$ by definitions. Therefore, 
	$\mathrm{span}_{\mathbb{C}}(S_W)\subseteq  \mathrm{span}_{\mathbb{C}}(H_W)\subseteq M_W.$ First, we show that  $M_W\subseteq \mathrm{span}_{\mathbb{C}}(H_W).$ In fact, for any $X\in M_W,$ we have $\tr[WX]=0.$ Therefore,  $\tr[WX^\dagger]=\tr[X^\dagger W]=\tr[X^\dagger W^\dagger]=\tr[(WX)^\dagger]=\tr[WX ]^\dagger=0.$ Then we have $ \tr[W(X+X^\dagger)]=0$ and $\tr[W(iX-iX^\dagger)]=0.$ However, since $H_1=(X+X^\dagger)$ and $H_2=i(X-X^\dagger) \in\mathbb{H}$, they both belong to $H_W$. Notice that $X=(H_1-iH_2)/2$. Hence it belongs to $\mathrm{span}_{\mathbb{C}}(H_W).$
	
	Now we prove that $\mathrm{span}_{\mathbb{C}}(H_W)\subseteq  \mathrm{span}_{\mathbb{C}}(S_W).$ It suffices to show that $H_W\subseteq  \mathrm{span}_{\mathbb{C}}(S_W).$ For any $H\in H_W,$  we have  $H\in\mathbb{H}$ and  $\tr[WH]=0.$ We can find small enough $\epsilon \in(0,1)$ such that $P_\epsilon:=\epsilon H+(1-\epsilon) \rho$ is positive. Set $N=\tr(P_\epsilon).$ Then $\rho_\epsilon:=P_\epsilon/N=\frac{\epsilon}{N} H+\frac{1-\epsilon}{N}\rho\in \mathbb{D}.$  Moreover, we have the equality,
	$$ \tr[W \rho_\epsilon]= \frac{\epsilon}{N} \tr[W H]+\frac{1-\epsilon}{N} \tr[W\rho]=0.$$
	Therefore, $\rho_\epsilon\in S_W.$   Note that $H$ can be written as a linear combination of $\rho_\epsilon$ and $\rho.$ Hence, $H_W\subseteq  \mathrm{span}_{\mathbb{C}}(S_W)$,
	which completes the proof. \end{proof}

\smallskip
\section{Proof of Theorem \ref{thm:On_Completeness}}\label{appen:thm1}

\begin{thm}[Completeness]\label{thm:On_Completeness}
	For each  $x\in \mathcal{X},$ the coherence criterion ($\mathbb{W}_x, D_x,I_x$) is complete.
\end{thm}
\begin{proof} Suppose $\rho$ is a coherent state. Without loss of generality, we assume $\rho_{mn}=\langle m|\rho|n\rangle \neq 0$, where $m\neq n$. We need to show that there always exist some $W\in\mathbb{W}_x$ such that $\tr[W\rho]\in  D_x \setminus I_x.$ As $x\in \mathcal{X}$, the interval $(-\infty,0)\subseteq  D_x \setminus I_x$. It is sufficient to prove that there always exist some $W\in\mathbb{W}_x$ such that $\tr[W\rho]<0$.
	We prove the theorem according to following four cases:
	
	Case (i): $x=0$. We define
	$ W_{j,k}^R:=(|j\rangle\langle k|+|k\rangle\langle j|)/2$ and $W_{j,k}^I:=\sqrt{-1}(|j\rangle\langle k|-|k\rangle\langle j|)/2$
	for each $1\leq j<k\leq d.$ Note that we always have
	$\tr[W_{j,k}^R\rho]=\Re(\rho_{jk})$ and $\tr[W_{j,k}^I\rho]=\Im(\rho_{jk})$,
	where $\rho_{jk}=\langle j|\rho|k\rangle$, $\Re(z)=\alpha$ and $\Im(z)=\beta$ for any complex number $z=\alpha+\sqrt{-1} \beta \in \mathbb{C}.$ Clearly, $\mathcal{W}_0:=\{W_{j,k}^R,W_{j,k}^I\ \big |\  1\leq j<k\leq d\} \subseteq \mathbb{W}_0^c.$  As $\rho_{mn}\neq 0$,   either $\tr[W_{m,n}^R\rho]\neq 0$ or  $\tr[W_{m,n}^I\rho]\neq 0.$   That is, $\rho$ is in   $\mathbb{C}_0[W_{m,n}^R]$ or $\mathbb{C}_0[W_{m,n}^I].$  Therefore,
	\begin{equation}\label{eq:finite_0}
		\mathbb{D}\setminus \Delta=  \bigcup_{W\in \mathcal{W}_0} \mathbb{C}_0[W].
	\end{equation}
	Without loss of generality, we assume that $\tr[W_{m,n}^R\rho]\neq 0.$ Since $\tr[W_{m,n}^R\rho]$ is a real number, $\tr[W_{m,n}^R\rho]$ or $\tr[-W_{m,n}^R\rho]$ is negative and both $\pm W_{m,n}^R$ are in $\mathbb{W}_0.$  Hence, there always exist some $W\in\mathbb{W}_0$ such that $\tr[W\rho]<0.$
	
	Case (ii): $x$ is $\geq$. Denote $\mathcal{W}_\geq :=\{\pm W_{j,k}^R,\pm W_{j,k}^I\ \big |\  1\leq j<k\leq d\}$. By definition, we have  $\mathcal{W}_\geq\subseteq \mathbb{W}_\geq^c.$  Similar to the argument in above case, there always exist some $W\in \mathcal{W}_\geq$ such that $\tr[W\rho]<0.$ Therefore,
	\begin{equation}\label{eq:finite_geq}
		\mathbb{D}\setminus \Delta=  \bigcup_{W\in \mathcal{W}_\geq} \mathbb{C}_\geq[W].
	\end{equation}
	
	Case (iii): $x$ is $>.$ From the case (ii), there is some $W\in \mathcal{W}_\geq$ such that  $\tr[W\rho]<0.$ Note that the diagonals of $W$ are all zeros. For each $\epsilon>0$, set $W_\epsilon=W+\epsilon \mathbb{I}_d$.  Clearly, $W_\epsilon$ has positive diagonals and $\tr[W_\epsilon]=d \epsilon>0$. Hence, $W_\epsilon\in \mathbb{W}_>$.   Specially, choosing $\epsilon =-\frac{\tr[W\rho]}{2}$ we have
	$$\tr[W_\epsilon \rho]=\tr[W\rho]+\epsilon \tr[\rho]=\tr[W\rho]+\epsilon=\frac{\tr[W\rho]}{2}<0.$$
	
	Case (iv): $x$ is a positive real number. By the case (iii), there is some $W_\epsilon\in \mathcal{W}_>$ such that $\tr[W_\epsilon\rho]<0.$ Set $W_x=\frac{x}{\tr[W_\epsilon]}W_\epsilon.$ Clearly, $W_x\in \mathbb{W}_x$ and $$\tr[W_x\rho]=\frac{x}{\tr[W_\epsilon]} \tr[W_\epsilon \rho]<0.$$
	
	The above four cases together complete the proof.
\end{proof}

\section{Proof of Theorem \ref{thm:Characterization_Completeness}}\label{appen:thm2}
\begin{thm}[Characterization of finite completeness]\label{thm:Characterization_Completeness}
	Let $x$ be an element in $\mathcal{X}.$ The coherence criterion ($\mathbb{W}_x, D_x,I_x$) is   finitely completable if and only if   $x$ is  $0$ or  $\geq$.
\end{thm}

\begin{proof} By Eqs. \eqref{eq:finite_0} and \eqref{eq:finite_geq} and the definition of finitely completable, both coherence criteria ($\mathbb{W}_0, D_0,I_0$) and  ($\mathbb{W}_\geq, D_\geq,I_\geq$) are finitely completable.
	
	Now we show that for any positive real number $x$, the criterion ($\mathbb{W}_x, D_x,I_x$) is   finitely incompletable. If not, there exists some $x>0$ and a set $\mathcal{W}:=\{W_i\}_{i=1}^n\subseteq \mathbb{W}_x^c$ such that
	$\mathbb{D}\setminus \Delta=\cup_{i=1}^n \mathbb{C}_x[W_i]$. That is, for any coherent states $\rho$, there always exists some $W\in \mathcal{W}$ such that
	$\tr[W\rho]<0$ or  $\tr[W\rho]>x.$
	On the other hand, set $\rho_\epsilon=\frac{\mathbb{I}_d}{d}+ \epsilon H$, where $H=|1\rangle\langle 2| +|2\rangle \langle 1|$ and $\epsilon \in (0,\frac{1}{d})$, and $M=\max\{\big|\tr[ (W-\Delta(W) H]\big|: W\in\mathcal{W}\}+1.$ Note that
	$$\tr[W\rho_\epsilon]=\frac{x}{d}+\epsilon \tr[(W-\Delta(W))H].$$
	If we set $\epsilon =\min\{ \frac{x}{2dM},\frac{1}{2d}\},$ then we have $0<\frac{1}{2d}x\leq  \tr[W\rho_\epsilon]\leq \frac{3}{2d}x<x$ for every $W\in \mathcal{W}.$ However, $\rho_\epsilon$ is a coherent state which can not be detected by all the witnesses in $\mathcal{W}.$
	
	Now we show that the criterion ($\mathbb{W}_>, D_>,I_>$) is also finitely incompletable. For any finite set $\mathcal{W}_>:=\{W_i\}_{i=1}^n\subseteq \mathbb{W}_>^c$, we define  $M_1:=\max\{\big|\tr[ (W-\Delta(W) H]\big|: W\in\mathcal{W}\}+1, $ and $m_1:=\min\{\tr[W]: W\in \mathcal{W}_>\}>0.$ If we set $\epsilon_1=\min\{\frac{m_1}{2dM_1},\frac{1}{2d}\},$ then $\tr[W\rho_{\epsilon_1}]\geq \frac{m_1}{2d}>0$ for every $W\in \mathcal{W}_>.$ That is,  $\rho_{\epsilon_1}$ is a coherent state whose coherence can not be detected by the witnesses in $\mathcal{W}_>.$
\end{proof}

\section{Proof of Theorem \ref{thm:Common_State}}\label{appen:thm3}
\begin{thm}[Property of finite intersection]\label{thm:Common_State}
	Let $x\in \mathcal{X}$ and $\mathcal{W}:=\{W_i\}_{i=1}^n\subseteq \mathbb{W}_x^c.$  Then   the following statements hold.
	\begin{enumerate}
		\item [{\rm(a)}]For the case $x\in (0,\infty)$ and any subset $\mathcal{W}'$ with cardinality $|\mathcal{W}'|\geq \lceil\frac{n}{2}\rceil$, if there always exist $t_W >0$ for each $W\in \mathcal{W}'$ and $\sum_{W\in \mathcal{W}'} t_W=1$ such that $  \sum_{W\in \mathcal{W}'}  t_W W \geq \boldsymbol{0}$, then  $ \bigcap_{i=1}^n \mathbb{C}_x[{W_i}] = \emptyset.$
		\item [{\rm(b)}] If $x$ is $>$ or $\geq,$  then
		$ \bigcap_{i=1}^n \mathbb{C}_x[{W_i}] = \emptyset
		$ if and only if there is some nonnegative  $t_i\geq 0$ and $\sum_{i=1}^n t_i=1$ such that $\sum_{i=1}^n t_iW_i$ is positive semidefinite.
		\item [\rm(c)] If $x=0,$ there exists a common coherent state that can be detected by all of $W_i$ via criterion of type (C). That is, we always have
		$$
		\mathcal{C}_0[{\mathcal{W}}]:=\bigcap_{i=1}^n \mathbb{C}_0[{W_i}] \neq \emptyset.
		$$
		Moreover, the cardinality of $\mathcal{C}_0[{\mathcal{W}}]$ is infinity. That is, given any set of finitely many coherent witnesses,  they     share with infinitely many common coherent states.
	\end{enumerate}
\end{thm}
\begin{proof}
	(a). If $\rho\in \bigcap_{i=1}^n \mathbb{C}_x[{W_i}] $, then $\tr[W_i \rho]<0$ or $\tr[W_i \rho]>x.$ Let $\mathcal{I}_0:=\{i \ \big |\  \tr[W_i\rho]<0, \ 1\leq i\leq n\}$ and  $\mathcal{I}_x:=\{i \ \big |\  \tr[W_i\rho]>x, \ 1\leq i\leq n\}.$
	Without loss of generality, we assume the set $\mathcal{I}_x$ is larger than $\mathcal{I}_0$. Then correspondingly the cardinality must be greater or equal to $\lceil \frac{n}{2}\rceil.$  Define $\mathcal{W}':=\{W_i|i\in \mathcal{I}_x\}.$ By assumption, there  are  $t_i>0 $ for $i\in\mathcal{I}_x$  such that $\sum_{i\in \mathcal{I}_x} t_i=1$ and  $ W:=\sum_{i\in \mathcal{I}_x} t_i W_i\geq \boldsymbol{0}.$ Therefore, all eigenvalues of $W$ are nonnegative. Moreover, as $\tr[W]=x$, all the eigenvalues of $W$ are less or equal to $x$. Hence, $W\leq x\mathbb{I}_d$, which implies $\tr[W\rho]\leq x.$ On the other hand, for each $i\in \mathcal{I}_x,$ we have $\tr[W_i\rho]>x$, which leads to
	$$
	\tr[W\rho]=\sum_{i\in\mathcal{I}_x} t_i \tr[W_i\rho]>x \sum_{i\in\mathcal{I}_x} t_i =x,
	$$
	and thus a contradiction. Hence $ \bigcap_{i=1}^n \mathbb{C}_x[{W_i}]=\emptyset.$
	
	\noindent (b). The argument for the case that $x$ is $\geq$ can be directly referred to reference \cite{Wang21}. When $x$ is $>$, as the $D_>$ and $I_>$ are the same as $D_\geq$ and $I_\geq$, the proof goes similarly to the one given in \cite{Wang21}.
	
	\noindent (c). We prove the statement by induction on $n.$ First, we consider the case $n=2.$ As both $\mathbb{C}_0[{W_1}]$ and  $\mathbb{C}_0[{W_2}]$  are nonempty, we assume $\rho \in \mathbb{C}_0[{W_1}]$ and  $\sigma\in \mathbb{C}_0[{W_2}].$ We may assume that  $\rho,\sigma\notin \mathbb{C}_0[{W_1}]\cap \mathbb{C}_0[{W_2}]$, i.e.,  $\rho\notin \mathbb{C}_0[{W_2}]$  and  $\sigma\notin \mathbb{C}_0[{W_1}].$  Equivalently, we have $\tr[W_2\rho]=\tr[W_1\sigma]=0.$ Then for any $\lambda\in (0,1),$ the matrix $\pi_{\lambda} :=\lambda\rho+(1-\lambda)\sigma\in \mathbb{D}.$ Moreover,
	$\tr[W_1\pi_\lambda]=\lambda\tr[W_1\rho]\neq 0$ and $\tr[W_2\pi_\lambda]=(1-\lambda)\tr[W_2\sigma]\neq 0$.
	That is, for each $\lambda\in (0,1),$ the state $\pi_{\lambda}$ belongs to  $\mathbb{C}_0[{W_1}]\cap \mathbb{C}_0[{W_2}]$.
	Now suppose that the statement holds when cardinality of the set of coherence witnesses is equal to $n-1.$ Let $\mathcal{W}=\{W_1,W_2,\cdots,W_n\}.$ We define $\mathcal{W}_1=\{W_1, \cdots,W_{n-1}\}$ and  $\mathcal{W}_2=\{ W_2,\cdots,W_{n}\}.$ By induction,
	the following two sets
	$$\mathcal{C}_0[{\mathcal{W}_1}]:=\bigcap_{W\in\mathcal{W}_1} \mathbb{C}_0[{W}],\ \ \mathcal{C}_0[{\mathcal{W}_2}]:=\bigcap_{W\in\mathcal{W}_2} \mathbb{C}_0[{W}],$$
	are nonempty. Suppose $\rho\in \mathcal{C}_0[{\mathcal{W}_1}]$ and  $\sigma\in\mathcal{C}_0[{\mathcal{W}_2}],$ i.e., $\tr[W_i \ \rho]\neq 0$ and  $\tr[W_{i+1}\sigma]\neq 0$  for all $i=1,2,\cdots,(n-1).$ If $\tr[W_{n}\rho]\neq 0$ (or $\tr[W_{1}\sigma]\neq 0$), we have $\rho\in \mathcal{C}_0[{\mathcal{W}}]$ (or $\sigma\in \mathcal{C}_0[{\mathcal{W}}]$) which proves the nonempty of $\mathcal{C}_0[{\mathcal{W}}]$. Therefore, we might assume $\tr[W_{n}\rho]=0$  and  $\tr[W_{1}\sigma]= 0$. Similarly, for each $\lambda\in(0,1),$ we define $\pi_{\lambda} :=\lambda\rho+(1-\lambda)\sigma\in \mathbb{D}.$ Then we have $ \tr[W_1\pi_\lambda]=\lambda\tr[W_1\rho]\neq 0$ and $\tr[W_n\pi_\lambda]=(1-\lambda)\tr[W_n\sigma]\neq 0.$
	Moreover, for each integer $i\in[2,n-1],$ we have $a_i:=\tr[W_{i}\rho]\neq 0$ and  $b_i:=\tr[W_{i}\sigma]\neq 0.$ Define
	$f_i(\lambda):=\tr[W_i\pi_\lambda]=\lambda a_i+(1-\lambda) b_i=(a_i-b_i)\lambda+b_i$
	and $f(\lambda):=\prod_{i=2}^{n-1} f_i(\lambda).$ If $a_i=b_i$ for all $i\in[2,n-1],$ we have $\tr[W_i\pi_\lambda]=b_i\neq 0.$ Then $\pi_\lambda\in\mathcal{C}_0[{\mathcal{W}}].$ If not, $f(\lambda)$ is a nontrivial (not a constant) polynomial of $\lambda.$ Therefore, there are only finite $\lambda$'s, say $\lambda_j$ $(j=1,2,\cdots,N)$ such that $f(\lambda_j)=0.$ Therefore, for each $\lambda\in (0,1)\setminus \{\lambda_j\}_{j=1}^N,$ we always have  $f(\lambda)\neq 0.$ Hence $\tr[W_i\pi_\lambda]=f_i(\lambda)\neq 0$ for all integer $i\in[2,n-1]$, which also holds for $i=1,n.$ Therefore, they all belong to the set $\mathcal{C}_0[{\mathcal{W}}].$
	
	For each coherent state $\rho$, there exists witness $W_\rho$ which can not detect the coherence of $\rho$, i.e., $\rho\notin \mathbb{C}_0[{W_{\rho}}].$ If $\mathcal{C}_0[{\mathcal{W}}]$ is finite, suppose $\mathcal{C}_0[{\mathcal{W}}]=\{\rho_j\}_{j=1}^m$. Then one has
	$$
	(\cap_{i=1}^n \mathbb{C}_0[{W_i}]) \bigcap   (\cap_{j=1}^m \mathbb{C}_0[{W_{\rho_j}}])= \emptyset,
	$$
	which leads to a contradiction.
\end{proof}

\section{ Proof of Theorem \ref{thm:inclusion_State}}\label{appen:thm4}
\begin{thm}[Property of inclusion or identity]\label{thm:inclusion_State}
	Let $x\in \mathcal{X}$ and $W_1, W_2\in \mathbb{W}_x^c.$ The following statements hold.
	\begin{enumerate}
		\item [{\rm(a)}] If $x\in (0,\infty),$ then $\mathbb{C}_x[{W_1}]=\mathbb{C}_x[{W_2}]$ if and only if $W_1=W_2$ or $W_1+W_2=x \mathbb{I}_d$ (the statement is true only if $d=2$).
		\item [{\rm(b)}]  If  $x$ is $>$ or $\geq$, $\mathbb{C}_x[{W_1}]=   \mathbb{C}_x[{W_2}]$ if and only if there exists $r\in\mathbb{R}_+:=\{r\in\mathbb{R} \ \big |\ r>0\}$ such that  $W_2= r W_1.$  Moreover,  $\mathbb{C}_x[{W_2}]\subseteq    \mathbb{C}_x[{W_1}]$  if and only if  there exists $a\in \mathbb{R}_+$ and a positive semidefintie operator $P$ such that $W_2= a W_1+P$.
		\item [\rm(c)] If $x=0,$ then $\mathbb{C}_0[W_1]=\mathbb{C}_0[W_2] $ if and only if there exist some $r\in \mathbb{R}\setminus\{0\}$ such that $W_2=rW_1$ (we call such two witnesses $\mathbb{R}$ equivalent). If $W_1$ and $W_2$ are not $\mathbb{R}$ equivalent, then $\mathbb{C}_0[W_1]\not\subseteq \mathbb{C}_0[W_2]$ and $\mathbb{C}_0[W_2]\not\subseteq \mathbb{C}_0[W_1].$
	\end{enumerate}
\end{thm}
\begin{proof}
	We first prove the statements (b) and   (c).
	
	(b). For any Hermitian matrix $W\in \mathbb{H}$, define $S_W:=\{\rho\in \mathbb{D} | \ \tr[W\rho]=0\}$ and $M_W:=\{ X\in \mathrm{Mat}_{d}(\mathbb{C}) | \ \tr[W X]=0\}.$  Suppose that $\mathbb{C}_x[W_1]=\mathbb{C}_x[W_2].$ We claim that $S_{W_1}=S_{W_2}.$ Otherwise, without loss of generality, we assume that $\rho\in S_{W_1}$ and $\rho\notin S_{W_2}$, i.e., $\tr[W_1\rho]=0$ and $\tr[W_2\rho]\neq 0.$ If $\tr[W_2\rho]< 0,$ then $\rho\in  \mathbb{C}_x[W_2]$ and $\rho\notin  \mathbb{C}_x[W_1]$ which contradicts to the assumption. Therefore, $\tr[W_2\rho]> 0.$  Now choosing any $\sigma \in \mathbb{C}_x[W_1],$ we have $\tr[W_1\sigma]<0$ and $\tr[W_2\sigma]<0.$ For each $\epsilon \in(0,1),$ define $\rho_\epsilon = \epsilon \sigma+(1-\epsilon)\rho\in \mathbb{D}.$ Then for small enough $\epsilon$, we have $\tr[W_1\rho_\epsilon]<0$ but  $\tr[W_2\rho_\epsilon]>0$, which implies that $\rho_\epsilon \in \mathbb{C}_x[W_1]$ but $\rho_\epsilon \notin \mathbb{C}_x[W_2].$ Hence we obtain again a contradiction. Therefore, $S_{W_1}=S_{W_2}.$  By Lemma \ref{lemma:spaceEqual}, $M_{W_1}=M_{W_2}$ implies $W_2=cW_1$ for some nonzero complex number. With this at hand, it is  easy to show that $c$ is in fact a positive real number. The other direction of the first statement in (b) is obvious.
	
	Now we prove the second statement. Suppose that there are $a\in \mathbb{R}_+$  and  a  positive semidefintie operator $P$ such that $W_2=a W_1+P.$  For any $\rho\in \mathbb{C}_x[W_2],$ we have $0>\tr[W_2 \rho] =a\tr[W_1\rho]+\tr[P\rho].$ This implies that  $\tr[W_1\rho]<0$ as we always have $\tr[P\rho]\geq 0$ for positive semidefinite $P.$ Therefore, $\rho\in \mathbb{C}_x[W_1]$ and $\mathbb{C}_x[W_2]\subseteq \mathbb{C}_x[W_1].$
	
	Suppose that $\mathbb{C}_x[W_2]\subseteq \mathbb{C}_x[W_1].$
	We prove that there exist $a\in\mathbb{R}_+$ and a positive semidefinite operator $P$ such that $W_2=aW_1+P.$ Set $t_1=\tr[W_1]$ and $t_2=\tr[W_2].$  Clearly, $t_1,t_2\geq 0$ as $W_1,W_2\in\mathbb{H}_\geq $. We prove the conclusion according to the following four cases.
	\begin{enumerate}
		\item[\rm{(i)}] Both $t_1$ and $ t_2$ are positive. Note that $\mathbb{C}_x[W_1]=\mathbb{C}_x[W_1/{t_1}]$ and $\mathbb{C}_x[W_2]=\mathbb{C}_x[W_2/{t_2}].$ Following the proof of Lemma 1 and the Corollary 1 in Ref. \cite{Wang21b} (which is correct under the assumption that the traces of the observables are 1), we have
		$ W_2/{t_2}=a_1 W_1/{t_1}+Q$ for some positive $a_1$ and positive semidefinite $Q.$ Therefore, $W_2= aW_1+P$, where $a=a_1{t_2}/{t_1}$ and $P={t_2}Q$.
		\item[\rm{(ii)}] $t_1=0$ but $t_2>0.$ It is easily verified that
		$\mathbb{C}_x[W_2]\subseteq   \mathbb{C}_x[W_1+W_2]\subseteq  \mathbb{C}_x[W_1].$ Similarly, we have
		\begin{equation}\label{eq:plusQ}
			W_2/{t_2}=a_1 (W_1+W_2)/{t_2}+Q
		\end{equation}
		for some positive $a_1$ and  positive semidefinite $Q.$ From the trace of Eq. \eqref{eq:plusQ}, we obtain $1=a_1+\tr[Q].$ Therefore, $0<a_1\leq 1.$ Clearly, $a_1\neq 1,$ otherwise, $\tr[Q]=0$ implies that $Q=\mathbf{0}$ and $W_1=-t_2/a_1Q=\mathbf{0}$ is not a nontrivial coherent witness. Hence, $0<a_1< 1.$ Eq. \eqref{eq:plusQ} can be reexpressed as
		$$
		W_2=\frac{a_1}{1-a_1} W_1+\frac{t_2}{1-a_1} Q.
		$$
		\item[\rm{(iii)}] $t_1> 0$ but $t_2= 0.$ One easily shows that
		$\mathbb{C}_x[W_2]\subseteq   \mathbb{C}_x[W_1+W_2]\subseteq  \mathbb{C}_x[W_1].$ Similarly, one has
		\begin{equation}\label{eq:plusQ2}
			(W_1+W_2)/{t_1}=a_1 W_1/{t_1}+Q
		\end{equation} for some positive $a_1$ and positive semidefinite $Q.$ The trace of Eq. \eqref{eq:plusQ2} gives rise to $1=a_1+\tr[Q].$ Therefore, $0<a_1\leq 1.$ Clearly, $a_1\neq 1,$ otherwise, $W_2=t_1Q=\mathbf{0}$ can not be a nontrivial coherent witness. We can rewrite Eq. \eqref{eq:plusQ2} as $(1-a_1)W_1+W_2 =t_1 Q$ and obtain that
		$\mathbb{C}_x[W_2]\subseteq   \mathbb{C}_x[(1-a_1)W_1+W_2]\subseteq  \mathbb{C}_x[W_1]$. On the other hand, $\mathbb{C}_x[ t_1Q]=\emptyset$. Therefore, this case can not happen as $\mathbb{C}_x[W_2]\neq \emptyset.$
		
		\item[\rm{(iv)}] $t_1=t_2=0.$ For this case, we show that this could happen only when $W_2=aW_1.$ In fact, for small enough $\epsilon >0,$
		$\mathbb{C}_x[W_2+\epsilon \mathbb{I}_d]\subseteq   \mathbb{C}_x[W_2]\subseteq  \mathbb{C}_x[W_1].$ By the argument of case (ii), we have $a_\epsilon>0$ and  positive semidefinite operator $P_\epsilon$ such that $W_2+\epsilon \mathbb{I}_d= a_\epsilon W_1+P_\epsilon$, which implies that the diagonals of $P_\epsilon$ are all $\epsilon$.  Moreover, as $P_\epsilon$ is positive semidefinite, the  module of each offdiagonal elements must be smaller than $
		\epsilon$. The matrix $P_\epsilon$ converges to zero matrix (up to entries wise) when $\epsilon \rightarrow 0$. Taking the limit $\epsilon \rightarrow 0$, we have  $W_2=aW_1.$
	\end{enumerate}
	
	(c). If $W_2=rW_1$ for $r\neq 0,$  we always have  $\tr[W_2\rho]=r\tr[W_1\rho]$ for all $\rho\in\mathbb{D}.$ In this case, $\tr[W_2\rho]=0$ if and only if $\tr[W_1\rho]=0$. Hence $\mathbb{C}_0[W_2]=\mathbb{C}_0[W_1].$ Conversely, if $\mathbb{C}_0[W_1]=\mathbb{C}_0[W_2]$,  we have equivalently $S_{W_1}=S_{W_2}.$ By Lemma \ref{lemma:spaceEqual}, we get $M_{W_1}=M_{W_2}.$ By definition, $M_{W_1} $ is just the orthogonal complement space to the vector $W_1$. Hence the space has dimensional $d^2-1.$ So the vectors that are orthogonal to each elements in $M_{W_1}$ are just $\mathbb{C} W_1.$ Therefore, there is nonzero $c\in\mathbb{C}$ such that $W_2=cW_1.$ Since $W_2$ is Hermitian, we have $c\in\mathbb{R}.$
	
	We prove the second statement in (c) by contradiction. Without loss of generality, we may assume that $\mathbb{C}_0[W_1]\subseteq \mathbb{C}_0[W_2].$ Then we have $M_{W_1}\subseteq M_{W_2}.$ As both are linear space of the same dimension $d^2-1,$ they must be equal, i.e., $M_{W_1}=M_{W_2}.$ Similar to the above derivations, we deduce that $W_2$ and $W_1$ are $\mathbb{R}$ equivalent, which contradicts to the assumption.
	
	(a). If $x$ is positive real and $W\in \mathbb{W}_x^c$, we have $\mathbb{C}_x[W]= \mathbb{C}_\geq [W] \cup \mathbb{C}_\geq [x\mathbb{I}_d-W].$
	If $x\mathbb{I}_d-W$ is not positive semidefinite, then $\mathbb{C}_\geq [x\mathbb{I}_d-W]\neq\emptyset.$ In such setting, $\mathbb{C}_x[W]$ is a disjoint unions of two convex sets.
	
	Suppose that $\mathbb{C}_x[W_1]=\mathbb{C}_x[W_2]$. If  $\mathbb{C}_\geq [x\mathbb{I}_d-W_1]=\emptyset,$  we must have $\mathbb{C}_\geq [x\mathbb{I}_d-W_2]=\emptyset$, which yields $\mathbb{C}_\geq [W_1]=\mathbb{C}_\geq [W_2].$ By the case (b), $W_2= rW_1$ for some positive $r$. Taking into account the trace, we must have $r=1$. If $\mathbb{C}_\geq [x\mathbb{I}_d-W_1]\neq \emptyset,$ we have $\mathbb{C}_\geq [x\mathbb{I}_d-W_2]\neq \emptyset.$ Moreover, either $\mathbb{C}_\geq [W_1]=\mathbb{C}_\geq [W_2]$ or $\mathbb{C}_\geq [W_1]=\mathbb{C}_\geq [x\mathbb{I}_d-W_2].$ For the former case, we have shown that $W_1=W_2$. For the latter case, we also have $\mathbb{C}_\geq [W_2]=\mathbb{C}_\geq [x\mathbb{I}_d-W_1]$, form which we can deduce $(d-1)W_1=x\mathbb{I}_d-W_2$ and $(d-1)W_2=x\mathbb{I}_d-W_1.$
	If $d>2$, we have $W_1,W_2\propto \mathbb{I}_d$, which contradicts with $W_1,W_2\in \mathbb{W}_x^c.$ While for the case $d=2$, the two equations merge into $W_1+W_2=x\mathbb{I}_d.$
	
	Moreover, for $d=2$ one has $\tr[W_1\rho]+\tr[W_2\rho]=x$ for any $\rho\in\mathbb{D}.$ Hence, $\tr[W_1\rho]<0$ if and only if $\tr[W_2\rho]>x$, and $\tr[W_1\rho]>x$ if and only if $\tr[W_2\rho]<0.$ Therefore, $\mathbb{C}_x[{W_1}]=\mathbb{C}_x[{W_2}].$
\end{proof}

\end{document}